\documentclass[a4paper]{llncs}
\usepackage[centertags]{amsmath}
\usepackage{amsfonts}
\usepackage{amssymb}
\def\wt{{\rm wt}}
\def\ord{{\rm ord}}

\title{New quantum codes from self-dual codes over $\F_4$}

\author{Reza Dastbasteh \and Petr Lison\v{e}k}

\institute{Simon Fraser University\\
\tt{rdastbas@sfu.ca\ \ plisonek@sfu.ca}}

\begin{document}

\newcommand{\Tr}{{\rm Tr}}
\newcommand{\F}{\mathbb{F}}
\newcommand{\Z}{\mathbb{Z}}

\maketitle

\begin{abstract}
We present new constructions of binary quantum codes from quaternary linear Hermitian self-dual codes. Our main ingredients for these constructions are nearly self-orthogonal cyclic or duadic codes over $\F_4$. An infinite family of $0$-dimensional binary quantum codes is provided. We give minimum distance lower bounds for our quantum codes in terms of the minimum distance of their ingredient linear codes. We also present new results on the minimum distance 
of linear cyclic codes using their fixed subcodes.
Finally, we list many new record-breaking quantum codes obtained from our constructions.
\end{abstract}

\begin{keywords}
quantum code, duadic code, quadratic residue code, cyclic code, self-dual code, minimum distance bound
\end{keywords}

\section{Introduction}

Quantum error-correcting codes or simply quantum codes are applied to protect quantum information from corruption by noise (decoherence) on the quantum channel 
in a way that is similar to that of classical error-correcting codes. In this paper we work exclusively with {\em binary
	quantum codes;} throughout the paper we will
mostly skip the adjective ``binary'' for brevity. 
The parameters of a binary quantum code that encodes $k$ logical qubits into $n$ 
physical qubits and has minimum distance $d$ 
are denoted by $[[n,k,d]]$. 

Our objective is to present theoretical results
that permit construction of good quantum codes with $k=0$
(sometimes called 0-dimensional quantum codes), and use
these results in numerical searches for good codes
of length up to 241. Our theoretical results
rely on the well known connection
of classical quaternary linear codes and binary quantum codes;
in this connection the 0-dimensional quantum codes
correspond to Hermitian self-dual classical linear codes.
We pay special attention to a certain family
of  linear cyclic codes known as duadic codes,
which are then confirmed numerically to be good sources
of 0-dimensional  quantum codes. However, we also provide other constructions
that are applicable to more general cyclic codes as well.
To support our numerical searches for good quantum codes
we also prove new theoretical results on bounds for the minimum
distance of classical cyclic codes.

We survey the required background in Section~2.
In Section~3 we provide constructions of quantum
codes based on duadic codes. In Section~4 we give constructions of quantum
codes based on more general classical codes.
In Section~5 we prove new minimum distance bounds
for cyclic codes using the fixed subcodes.
In Section~6 we present our numerical results
which include 12 new record breaking 0-dimensional
quantum codes; many other record breaking
codes can be derived by secondary constructions.

\section{Background}

An important class of quantum codes are quantum stabilizer codes. Binary stabilizer codes were
introduced  by Calderbank et al.\ \cite{Calderbank} and Gottesman \cite{Gottesman}. Each binary stabilizer code is a quaternary additive code (an additive subgroup of $\F_4^n$) which is dual containing with respect to a certain trace inner product \cite{Calderbank}.
For more information about the structure of quantum codes and their 
algebraic constructions we refer to 
the recent survey \cite{Grassl2}.
In this paper, we only restrict our attention to {$\F_4$-}linear subspaces of $\F_4^n$, and the following theorem gives the connection between quaternary linear codes and quantum codes.
\begin{theorem}\cite[Theorem 2]{Calderbank}\label{quantum def}
  Let $C$ be a linear $[n,k,d]$ code over $\F_4$ such that $C^{\bot_h}\subseteq C$, where $C^{\bot_h}$ is the Hermitian dual of $C$. Then we can construct an $[[n, 2k-n,d']]$ binary quantum code, where $d'$
 is the minimum weight in $C\setminus C^{\bot_h}$.
 If $C= C^{\bot_h}$ then $d'=d$.
\end{theorem}
If the quantum code of Theorem $\ref{quantum def}$ has minimum distance $d'=d$, then the code is called a {\em pure} quantum code.

There are several secondary constructions for quantum codes which take a quantum code and produce new quantum codes after applying standard constructions such as puncturing, lengthening, and shortening of the original code. The next theorem provides two such constructions.

\begin{theorem}\cite[Theorem 6]{Calderbank}\label{secondary}
Suppose that an $[[n,k,d]]$ quantum code exists. 
\begin{enumerate}
\item If $n \geq 2$ and the code is pure, then there exists an $[[n-1,k+1,d-1]]$ quantum code.
\item If $n \geq 2$, then an $[[n-1,k,d-1]]$ quantum code exists. 
\end{enumerate}
 \end{theorem}

A new secondary construction for CSS quantum codes
was recently discovered
by M.~Grassl. 
In contrast to case~2 of Theorem~\ref{secondary} above,
it allows one to decrease the length of a code by 2 while
decreasing its distance only by~1.
The details are given in \cite{Grassl3}.

\subsection*{Duadic codes}
\label{subsec-duadic}

{\em Duadic codes} are an important class of linear cyclic codes and they are
thoroughly discussed in \cite[Chapter 6]{Huffman} and \cite[Section 2.7]{Encyclopedia}. We briefly recall several important properties of this class of linear codes below.

Let $q$ be a prime power and $\F_q$ be the field of $q$ elements. Throughout the rest of this section, $n$ is a positive integer such that $\gcd(n,q)=1$. 

A linear code $C\subseteq \F_q^n$ is called \textit{cyclic} if for every $c=(c_0,c_1,\cdots,c_{n-1})\in C$, the vector $(c_{n-1},c_0,\cdots,c_{n-2})$ obtained by a cyclic shift of the coordinates of $c$ is also in $C$. It is well known that there is a one-to-one correspondence between cyclic codes of length $n$ over $\F_q$ and ideals of the ring $\F_q [x]/\langle x^n-1\rangle $, for example see \cite[Section 4.2]{Huffman}. Under this correspondence, each cyclic code can be uniquely represented by a monic polynomial $g(x)$, where $g(x)$ is the minimal degree generator of the corresponding ideal. The polynomial $g(x)$ is called the {\em generator polynomial} of such cyclic code. Let $\alpha$ be a fixed
primitive $n$-th root of unity in $\F_{q^r}$, 
where $r$ is the smallest positive integer such that
$n|(q^r-1)$.
 Alternatively, we can represent the above cyclic code by its unique {\em defining set}
$$\{t: 0\leq t \le n-1 \ \text{and}\ g(\alpha^t)=0 \}.$$
  
{\em Remark.}
In the numerical examples of cyclic codes throughout
this paper,  the $n$-root of unity $\alpha$ is fixed as follows.
Let $\gamma$ be the primitive element in $\F_{q^r}$
chosen by Magma \cite{magma}, 
then set $\alpha=\gamma^{(q^r-1)/n}$.

For each $a\in \mathbb{Z}_n$, the set $Z(a)=\{(aq^j) \bmod n: 0\le j \le m-1\}$, where $m$ is the smallest positive integer such that $aq^m \equiv a \pmod{n}$ is called a \textit{$q$-cyclotomic coset} modulo $n$. The $q$-cyclotomic cosets partition $\mathbb{Z}_n$ and each defining set
of a linear cyclic code is a union of cyclotomic cosets.
 
\begin{lemma}\label{e lisonek}
Let $n$ be a positive odd integer and $C$ be a length $n$ linear cyclic code over $\F_4$ with the defining set $A$. Then
$C^{\bot_h}$ has defining set $\Z_n \setminus (-2A)$.
Moreover, $C^{\bot_h} \subseteq C$ if and only if $A \cap -2A=\emptyset$.
\end{lemma}

For any integer $a$ such that $\gcd(n,a)=1$, the function $\mu_a$ defined on $\mathbb{Z}_n$ by $\mu_a(x)=(ax) \bmod {n}$ is called a {\em multiplier}. Clearly a multiplier is a permutation of $\mathbb{Z}_n$.

\begin{definition}\cite[Section 2.7]{Encyclopedia} 
Let $S_1$ and $S_2$ be unions of $q$-cyclotomic cosets modulo $n$ such that 
\begin{enumerate}
\item $0 \notin S_1 \cup S_2$
\item $S_1 \cup S_2 \cup \{0\}=\mathbb{Z}_n$ and $S_1 \cap S_2= \emptyset$, 
\item there is a multiplier ${\mu}_b$ such that ${\mu}_bS_1=S_2$ and ${\mu}_bS_2=S_1$.
\end{enumerate}
Then the pair $\{S_1,S_2\}$ is called a splitting of $\Z_n$ given by $\mu_b$ over $\F_q$. 
\end{definition}

A vector $(x_1,x_2,\cdots,x_n)\in \F_q^n$ is called {\em even-like} provided that $\displaystyle \sum_{i=1}^{n}x_i=0$
and it is called {\em odd-like} otherwise. A linear code is called {\em even-like} if it has only even-like codewords; a linear code is called {\em odd-like} if it is not even-like. In the binary case an even-like code has only even weights. 

Binary duadic codes were first introduced by Leon et al.\ \cite{Leon-Pless}, and later they were generalized to larger fields by Pless \cite{Pless,Pless3}.
 
\begin{definition}[Duadic codes]\cite[Theorem 6.1.5]{Huffman} \cite[Section 2.7]{Encyclopedia} 
Let $\{S_1,S_2\}$ be a splitting of $\Z_n$ over $\F_q$.
Then the linear cyclic codes with the defining sets $S_1 \cup \{0\}$ and $S_2 \cup \{0\}$ are called a pair of {even-like duadic} codes. The linear cyclic codes with the defining sets $S_1$ and $S_2$ are called a pair of {odd-like duadic} codes.
\end{definition}

 A comprehensive list of important properties of duadic codes is provided below. 
 
\begin{theorem}\cite[Theorem 6.1.3]{Huffman} \label{duadic properties}
Let $(C_1,C_2)$ and $(D_1,D_2)$ be pairs of even-like and odd-like duadic codes of length $n$ over $\F_q$, respectively, such that $C_1\subseteq D_1$ and $C_2 \subseteq D_2$. Then
\begin{enumerate}
\item $C_1$ and $C_2$ (respectively $D_1$ and $D_2$) are permutation equivalent codes. 
\item $C_1\cap C_2=\{0\}$ and $C_1+C_2$ is the cyclic code generated by $x-1$.
\item $D_1 \cap D_2=H$ and $D_1+D_2=\F_q^n$, where $H$ is the subspace of $\F_q^n$ spanned by the all-ones vector. 
\item $\dim C_1=\dim C_2=(n-1)/2$ and $\dim D_1=\dim D_2=(n+1)/2$.
\item $C_1$ is the subcode of $D_1$ containing all even-like vectors. The same holds for $C_2$ as the subcode of $D_2$.
\item $D_1=C_1\oplus H$ and $D_2=C_2\oplus H$.
\item
If $C_1$ is Hermitian self-orthogonal,
then $C_1^{\perp_h}=D_1$ and $C_2^{\perp_h}=D_2$.
\end{enumerate}
\end{theorem}

 Now we briefly mention the class of quadratic residue codes which are special cases of duadic codes. Let $p$ be an odd prime number. Let $Q_p$ be the set of non-zero squares (quadratic residues) modulo $p$ and $N_p$ be the set of nonsquares (quadratic nonresidues) modulo $p$. 
 The sets $Q_p$ and $N_p$ satisfy the following properties: 
 \begin{enumerate}
 \item $|Q_p|=|N_p|=\frac{p-1}{2}$.
 \item $aQ_p=Q_p$ and $aN_p=N_p$ for any $a\in Q_p$. Also, $bQ_p=N_p$ and $bN_p=Q_p$ for any $b\in N_p$. 
\end{enumerate}

If $q\in Q_p$ then each 
$q$-cyclotomic coset modulo $p$ 
different from $\{0\}$
either is a subset of $Q_p$ or
it is a subset of $N_p$.
Thus $Q_p$ and $N_p$ give a splitting of $\Z_p$ given by $\mu_{b}$ for any $b\in N_p$. The duadic codes corresponding to such splitting are called {\em quadratic residue codes,}
abbreviated 
{\em QR codes,} of length $p$ over $\F_q$.

Self-orthogonal duadic codes
and QR codes over $\F_4$ with respect to 
the Hermitian inner products are discussed below. 

\begin{theorem}\cite[Theorem 6.4.4]{Huffman}\label{duadic}
Let $C$ be a linear cyclic code over $\mathbb{F}_4$ with parameters $[n,\frac{n-1}{2}]$. Then $C$ is Hermitian self-orthogonal if and only if $C$ is an even-like duadic code with the multiplier ${\mu}_{-2}$. 
\end{theorem}

\begin{theorem}\cite[Section 6.6.1]{Huffman}\label{QR prop}
Let $p$ be an odd prime. 
The even-like QR codes of length $p$ over $\F_4$ are Hermitian self-orthogonal
 if and only if
 $p\equiv-1 \pmod 8$ or $p\equiv-3 \pmod 8$.
\end{theorem}


Let $D$ be an odd-like duadic code with the even-like subcode $C$. The {\em minimum odd-like weight} of $D$ is defined by $$d_o=\min\{{\rm wt}(v):v\in D\setminus C\}.$$ Several minimum distance conditions for duadic and QR codes are provided below. 
Let $d(C)$ denote the minimum distance of $C$.

\begin{theorem} \cite[Theorems 6.5.2, 6.6.6, and 6.6.22]{Huffman}\label{square root bound}
Let $D$ be an odd-like duadic code of length $n$ over $\F_q$. Let $d_o$ be the minimum odd-like weight of $D$. Then 
\begin{enumerate}
\item $d_o^2 \geq n$. 
\item If the splitting is given by $\mu_{-1}$, then $d_o^2-d_o+1\geq n$. 
\item
Furthermore, if $n$ is a prime number and $D$ is a QR code, then
\begin{enumerate}
\item[a.] $d(D)=d_o$.
\item[b.] If $q=2$ or $q=4$ and $n\equiv-1 \pmod 8$, 
then $d(D)\equiv 3 \pmod 4$. 
\end{enumerate}
 \end{enumerate}
\end{theorem}

An extended version of this result is provided in \cite[Theorems 6.5.2, 6.6.22]{Huffman}. In general, although the square root bound is a nice theoretical result, our computations given in Table $\ref{Table2}$ show that it does not provide a tight bound for the minimum distance.

We conclude this section with some useful information regarding when a splitting over $\F_4$ is given by $\mu_{-2}$, or in other words when a duadic code over $\F_4$ is Hermitian self-orthogonal by Theorem $\ref{duadic}$. 

%
\begin{theorem}\cite[Theorems 6.4.9 and 6.4.10]{Huffman}\label{duadic splitting}
Let $p$ be an odd prime number.
\begin{enumerate}
\item If $p\equiv -1 \pmod{8}$ or $p\equiv-3 \pmod{8}$, then every splitting of $\Z_p$ over $\F_4$ is given by $\mu_{-2}$.
\item If $p\equiv 3 \pmod 8$, then there is no splitting of $\Z_p$ given by $\mu_{-2}$ over $\F_4$.
\item If $p\equiv 1 \pmod{8}$, then $\mu_{-2}$ may or may not give a splitting of $\Z_p$ over $\F_4$. 
\end{enumerate}
Moreover, if $\mu_{-2}$ and $\mu_{-1}$ give the same splitting of $\Z_p$ over $\F_4$, then
$p\equiv \pm 1 \pmod{8}$. In particular, if $p\equiv -1 \pmod{8}$, then $\mu_{-2}$ and $\mu_{-1}$ give the same splitting of $\Z_p$ over $\F_4$.  
\end{theorem}

\section{A new class of good binary quantum codes}\label{construction section}
A $0$-dimensional quantum code with length $n$ has parameters $[[n,0,d]]$. Such a quantum code represents a single quantum state capable of correcting any $(d-1)/2$ errors. In practice, $0$-dimensional quantum codes can be useful for example in testing whether certain storage locations for qubits are decohering faster than they should \cite{Calderbank}. Moreover, higher-dimensional quantum codes can be constructed by applying Theorem $\ref{secondary}$ part 1 to a $0$-dimensional quantum code.

In this section, we provide a new infinite family of $0$-dimensional quantum codes using duadic codes over $\F_4$. Our construction targets nearly self-orthogonal duadic codes and also bounds the minimum distance of the constructed quantum code using minimum distances of an odd-like and an even-like duadic code. Throughout this section, $n$ always is a positive odd integer. For any integer $a$ such that $\gcd(a,n)=1$, we denote the multiplicative order of $a$ modulo $n$ by 
${\rm ord}_n(a)$.

Constructions of $1$-dimensional quantum codes
 can be found in the literature. 
 One such construction is provided below which is obtained by applying the CSS construction to binary duadic codes. 
 
\begin{theorem}\cite[Theorems 4 and 10]{aly2006}\label{Aly first}
Let $n$ be a positive odd integer. Then there exists a quantum code with parameters $[[n,1,d]]$, where $d^2\geq n$. If ${\rm ord}_n(2)$ is odd, then $d^2-d+1\geq n$.
\end{theorem} 
Moreover, Guenda in \cite{guenda} proved that the distance bound $d^2-d+1\geq n$ in Theorem \ref{Aly first} is still valid when ${\rm ord}_n(4)$ is odd. She also found the following new family of quantum codes when ${\rm ord}_n(4)$ is even.
\begin{theorem}\cite[Theorem 16]{guenda}
Let $n=p^m$ be a prime power power, $\gcd(p,2)=1$, and
${\rm ord}_n(4)$ be even. Then there exists an $[[n, 1, d]]$ quantum code with $d^2\ge n$.
\end{theorem}

For $u,v\in\F_4^n$ let
$\langle u,v\rangle_h$ denote their Hermitian inner product.
Further let $\lVert u \rVert=\langle u,u\rangle_h$. One can easily see that $\lVert u \rVert =\wt(u) \mod 2$.
The next theorem gives some useful information about the weights in certain even-like and odd-like quaternary duadic codes.
\begin{lemma}\label{odd weight duadic}
Let $n$ be a positive odd integer and $C_o$ be an odd-like duadic code 
of length $n$
with the multiplier $\mu_{-2}$ over $\F_4$. Let $C_e \subseteq C_o$ be the Hermitian dual of the code $C_o$. Then all codewords in $C_e$ have even weights and all codewords in $C_o\setminus C_e$ have odd weights. 
\end{lemma}  
\begin{proof}
Since $C_e$ is Hermitian self-orthogonal, for any $v \in C_e$,  we have $\lVert v \rVert=0$.  
This proves the first part. 

Let $j$ be the all-ones vector of length $n$ and $H$ be the subspace spanned by $j$ over $\F_4$. By Theorem $\ref{duadic properties}$ part 6, $C_o=C_e\oplus H$. 
Let $u+\alpha j$ be an arbitrary element of $C_o\setminus C_e$, where  $u \in C_e$ and $0 \neq\alpha \in \F_4$. Then $$\lVert u+\alpha j \rVert =\lVert u \rVert+ \lVert \alpha j \rVert+ \langle u, \alpha j \rangle_h +\langle \alpha j,u \rangle_h=1.$$
Hence $u+\alpha j$ has an odd weight.
$\hfill\square$
\end{proof}

%

We recall the following construction of quantum codes from linear codes which is called the {\em nearly self-orthogonal construction of quantum codes} \cite{Lisonek}. This construction extends a linear code, which is not necessarily Hermitian dual containing, to a Hermitian dual containing linear code of a larger length. 

\begin{theorem}\cite{Lisonek}\label{lisonek duadic}
	Let $C$ be an ${[n,k]}$ linear code over $\F_4$ and $e=n-k- \dim(C \cap C^{\bot_h})$. Then there exists a
	quantum code with parameters $[[n+e,2k-n+e,d]]$, where 
	$$ d\geq \min\{d(C), d(C+ C^{\bot_h}) +1\}.$$ 
\end{theorem}

It will be useful to introduce a name and provide a meaning
for the parameter~$e$ occurring in Theorem~\ref{lisonek duadic}.
We define the {\em near self-orthogonality} of a linear code $E$ with respect to the Hermitian inner product by 
$\dim(E)-\dim(E\cap E^{\bot_h})$.
This non-negative integer is~$0$ if and only if $E$ is self-orthogonal,
and in general it measures by how much $E$ misses
the self-orthogonality condition.
Thus
the meaning of the parameter $e$ in Theorem~\ref{lisonek duadic}
is the {\em near self-orthogonality of the code $C^{\perp_h}$,}
which equivalently measures by how much $C$ misses the 
dual-containment condition required in Theorem~\ref{quantum def}.

Next, we classify all the odd-like duadic codes having the near self-ortho\-gona\-lity $e=1$ with respect to the Hermitian inner product.

\begin{theorem}\label{duadic e=1}
Let $C$ be an odd-like duadic code. Then $C$ has the near self-orthogonality parameter $e=1$ if and only if $C$ has multiplier $\mu_{-2}$.
\end{theorem}

\begin{proof}
First suppose that $\mu_{-2}$ is a multiplier of $C$. Thus there exists a splitting of $\Z_n$ given by $\mu_{-2}$ in the form $(S_1,S_2)$ such that $S_1$ is the defining set of $C$. The code $C^{\bot_h}$ has the defining set $\Z_n \setminus (-2S_1)=\Z_n \setminus S_2=S_1\cup \{0\}$. Hence $C^{\bot_h}$ is the even-like duadic subcode of $C$ and $$e=\dim(C)-\dim(C\cap C^{\bot_h})=1.$$

Conversely let $(S_1',S_2')$ be a splitting of $\Z_n$ given by $\mu_{a}$ and $C$ be an odd-like duadic code with the defining set $S_1'$ and 
assume that $e=1$. Then

\begin{equation}\label{e=1 duadic}
\begin{split}
e=\dim(C)-\dim(C\cap C^{\bot_h})&=n-|S_1'|-\Big(n-|S_1' \cup \big(\Z_n\setminus (-2S_1')\big)| \Big)\\&=| S_1' \cup \big(\Z_n\setminus (-2S_1')\big)| -| S_1'|.
\end{split}
\end{equation}

Now if $-2S_1'\neq S_2'$, then $\{0,s\} \subseteq \Z_n\setminus (-2S_1')$ for some $s\in S_2'$. Thus $(\ref{e=1 duadic})$ implies that $e\geq 2$ which is a contradiction. Therefore, $-2S_1'= S_2'$ and $\mu_{-2}$ is a multiplier of $C$.
$\hfill\square$
\end{proof}

Next we construct a new family of $0$-dimensional quantum codes. In fact the quantum codes that we are constructing in this section are extended odd-like duadic codes. Later, in Section $\ref{section4}$, we provide a generalization of our construction to any linear codes over $\F_4$. In spite of the known theoretical results on duadic codes and their extended codes, they are not computationally discussed much in the literature. For instance, in \cite{Huffman}, the parameters of length $n$ duadic codes over $\F_4$ are only stated for $n \le 41$. Hence we take advantage of this opportunity and compute the parameters of extended duadic codes for much larger lengths ($n \le 241$).
Now, we state our main result of this section.

\begin{theorem}\label{duadic-quantum}
Let $n$ be a positive odd integer and $C_o$ be an odd-like duadic code of length $n$ with the multiplier $\mu_{-2}$ over $\F_4$. Then there exists a binary quantum code with parameters $[[n+1,0,d]]$, where 
\begin{enumerate}
\item $d \geq \min \{d(C_e),d(C_o)+1\}$, where $C_e$ is the even-like subcode of $C_o$.
\item $d$ is even.
\item If $d(C_o)$ is odd, then $d\geq \sqrt{n}+1$. Moreover, if also $\mu_{-1}$ is a multiplier for $C_o$, then $d^2-3(d-1)\geq n$. 
\end{enumerate}

\end{theorem}

\begin{proof}
Let $(S_1,S_2)$ be a splitting of $\Z_n$ given by $\mu_{-2}$ over $\F_4$ and $C_o$ and $C_e$ be the odd-like and even-like duadic code with the defining sets $S_1$ and $S_1\cup\{0\}$, respectively. By Theorem $\ref{duadic e=1}$, the code $C_o$ has the near self-orthogonality $e=1$.

The code $C_e$ has parameters $[n,\frac{n-1}{2}]$. Now applying the quantum construction given in Theorem $\ref{lisonek duadic}$ to $C_e$ results in an Hermitian self-dual linear code $Q$ which is also a quantum code with parameters $[[n+1,0,d]]$, where $d \geq \min\{d(C_e), d (C_o) +1\}$. The facts that $Q$ is linear and Hermitian self-dual imply that all weights
in $Q$ are even, as was shown in the proof
of Lemma \ref{odd weight duadic}.

Note that Lemma $\ref{odd weight duadic}$ implies that if
$d(C_o)$ is odd,
then $d(C_o)=d_o$ and $d_o < d(C_e)$. Thus $d_o$ satisfies the square root bound given in Theorem $\ref{square root bound}$.
 The facts that $d\geq d_o+1$ and $d_o\ge \sqrt{n}$ show that $d\geq \sqrt{n}+1$.

Finally, if the same splitting is given by $\mu_{-1}$ and $d(C_o)=d_o$ is odd, then by Theorem $\ref{square root bound}$, $d_o^2-d_o+1\geq n$. Now combining $d-1\geq d_o$ with the previous inequality gives the result.  
$\hfill\square$
\end{proof}
 
The lower bound that we provided in case 1 of Theorem~\ref{duadic-quantum} 
appears to be very good and almost all of our computational results rely on this lower bound. 

Restricting the code lengths to prime numbers in the form $p\equiv-1 \pmod{8}$ or $p\equiv-3 \pmod{8}$ leads to an infinite family of $0$-dimensional quantum codes of length $p+1$.

\begin{corollary}\label{duadic construction 2}
Let $p$ be a prime number such that $p\equiv-1 \pmod{8}$ or $p\equiv-3 \pmod{8}$. Then there exists a $[[p+1,0,d]]$ quantum code with an even minimum distance $d$ and
$$d \geq \min \{d(C_e),d(C_o)+1\},$$ where
$C_o$ is an odd-like duadic code 
over $\F_4$ of length $p$ with multiplier $\mu_{-2}$,
and $C_e$ is the even-like subcode of $C_o$. If $C_o$ is a QR code then $d\ge d(C_o)+1$.
Finally, if $C_o$ is a QR code, $p\equiv-1 \pmod 8$, and $d=d(C_o)+1$, then $d\equiv0\pmod 4$.
\end{corollary}

\begin{proof}
	We note that
by Theorem $\ref{duadic splitting}$ part 1
the only possible multiplier is $\mu_{-2}$.
The first part
of the claim follows from Theorem $\ref{duadic-quantum}$.
If $C_o$ is a QR code, then Theorem $\ref{square root bound}$ part 3a implies that $d(C_o)=d_o$. Moreover, Lemma $\ref{odd weight duadic}$ implies that $d_o$ is odd and $d(C_e)$ is even. Thus $d_o<d(C_e)$ and $d \geq \min \{d(C_e),d(C_o)+1\}$ implies that $d \geq d_o+1$.
The last fact about the minimum distance follows from Theorem $\ref{square root bound}$ part 3b which implies that $d(C_o)\equiv3\pmod 4$.
$\hfill\square$
\end{proof}

For each positive odd integer $n$, we have $\ord_n(4)\mid \ord_n(2)$ and if $\ord_n(2)$ is odd, then $\ord_n(4)= \ord_n(2)$.
 In the latter case, the binary and quaternary cyclotomic cosets modulo $n$ are the same. 
 Thus the binary and quaternary duadic codes have the same defining sets. In this special case, the following result allows one to compute the minimum distance of quaternary duadic codes much faster by only using the binary duadic code with the same defining set. 
 
\begin{theorem} \cite[Theorem 4]{Pless}\label{binary-quaternary}
Let $C$ be a quaternary linear code of minimum distance $d$ which is generated by a set of binary vectors. Then the binary linear code generated by the same set of generators has the minimum distance $d$.  
\end{theorem} 

Although Theorem $\ref{binary-quaternary}$ is stated for linear codes over $\F_4$, in general it holds for linear codes over any finite field extension of the binary field; see Theorem 3.8.8 of \cite{Huffman}.
Next, we give an analogue of Theorem $\ref{lisonek duadic}$ to binary cyclic codes that satisfy Theorem $\ref{binary-quaternary}$. We denote the Euclidean dual of a binary linear code $C$ by $C^\bot$.

\begin{corollary}
Let $n$ be a positive odd integer such that $\ord_n(4)=\ord_n(2)$. If $C$ is an $[n,k]$ binary cyclic code and 
$e=n-k-\dim(C \cap C^\bot)$, then there exists a quantum code with parameters $[[n+e,2k-n+e,d]]$, where $d \geq \min\{d(C), d(C+C^\bot)+1\}$.
\end{corollary}

\begin{proof}
First note that since $\ord_n(4)=\ord_n(2)$, the $2$-cyclotomic and $4$-cyclotomic cosets are the same modulo $n$. Hence binary and quaternary cyclic codes of length $n$ with a fixed defining set have the same dimension over $\F_2$ and $\F_4$, respectively. Let $A$ be the defining set of $C$, and let $D$ be the linear cyclic code over $\F_4$ with the defining set $A$. Thus $D$ is an $[n,k]$ linear code over $\F_4$. The defining set of $D^{\bot_h}$ is $\Z_n \setminus (-2 A)=\Z_n \setminus (-A)$, where the last equality follows from the fact that $2A=A \pmod n$. Hence $D^{\bot_h}$ and $C^\bot$ have the same defining sets. A similar argument shows that $C\cap C^\bot$ (respectively $C+C^\bot$) and $D \cap D^{\bot_h}$ (respectively $D+D^{\bot_h}$) have the same defining sets.
Therefore, $\dim(C\cap C^\bot)=\dim(D\cap D^{\bot_h})$ as linear codes over $\F_2$ and $\F_4$, respectively. Finally, Theorem $\ref{binary-quaternary}$ implies that $d(C^\bot)=d(D^{\bot_h})$ and $d(C+C^\bot)=d(D+D^{\bot_h})$.
Now the result follows by applying Theorem $\ref{lisonek duadic}$ to the code $D$. 
$\hfill\square$
\end{proof}
 
One of advantages of the above result is that binary duadic codes have been studied extensively in the literature. For instance, the exact or probable minimum distance of all binary duadic codes of length $n\le 241$ are determined in \cite{Smid}, \cite{Pless2}, and \cite[Section 6.5]{Huffman}.

\section{A more general family of 0-dimensional quantum codes}\label{section4}

In this section, we generalize the result of Theorem $\ref{duadic-quantum}$. We provide a method of constructing Hermitian self-dual codes over $\F_4$, or equivalently $0$-dimensional quantum codes. We will provide three examples of record-breaking
$0$-dimensional quantum codes obtained from nearly self-orthogonal linear codes over $\F_4$ with $e=3$.

\begin{theorem}\label{general quantum construction}
Let $C$ be an $[n,k]$ linear code over $\F_4$ such that $C\subseteq C^{\bot_h}$. Then there exists a quantum code with parameters $[[2(n-k),0,d]]$, where $d$ is even and $d \ge \min \{d(C),d(C^{\bot_h})+1\}$. 
\end{theorem}

\begin{proof}
We apply Theorem $\ref{lisonek duadic}$ to the code $C$. We get $e=n-k-\dim(C \cap C^{\bot_h})=n-2k$. Now, Theorem $\ref{lisonek duadic}$ implies the existence of an $[[2(n-k), 0,d ]]$ quantum code which satisfies $d \ge \min \{d(C),d(C^{\bot_h})+1\}$.
Such quantum code is also an Hermitian self-dual code over $\F_4$. Since Hermitian self-dual codes over $\F_4$ only have even weights, $d$ is even.
$\hfill\square$
\end{proof}

Note that although we only stated the result of Theorem \ref{general quantum construction} for quantum codes,
we also obtain  a $[2(n-k),n-k,d]$ Hermitian self-dual code over $\F_4$, where $d$ is even and $d \ge \min \{d(C),d(C^{\bot_h})+1\}$.

Theorem $\ref{general quantum construction}$ implies the following secondary construction of quantum codes.

\begin{corollary}\label{new secondary}
Let $C$ be a linear
code over $\F_4$ which is also an $[[n,2k-n]]$ quantum code
(in the sense of Theorem~\ref{quantum def}). Then there exists a quantum code with parameters $[[2k,0,d]]$, where $d$ is even and $d\geq \min \{d(C^{\bot_h}),d(C)+1\}$.
\end{corollary}

\begin{proof}
First note that $C$ is an $[n,k]$ linear code over $\F_4$ and $C^{\bot_h}\subseteq C$. Now applying Theorem $\ref{general quantum construction}$ to the code $C^{\bot_h}$ implies the existence of a $[[2k,0,d]]$ quantum code such that $d$ is even and $d\geq \min \{d(C^{\bot_h}),d(C)+1\}$. 
$\hfill\square$
\end{proof}

\begin{example}
The best known quantum code with parameters $[[93,3]]$ is in correspondence with a Hermitian dual containing linear code $C$ over $\F_4$. Then Corollary $\ref{new secondary}$ implies the existence of a $[[96,0]]$ quantum code. Moreover, $d(C)=21< d(C^{\bot_h})$ since $C^{\bot_h}\subset C$ and $C^{\bot_h}$ has an even weight. Hence there exists a {\em new quantum code} with parameters $[[96,0,22]]$. The previous best known quantum code with the same length and dimension had minimum distance $20$.
\end{example}

Next, we apply Theorem $\ref{general quantum construction}$ to linear cyclic codes over $\F_4$. 

\begin{corollary}\label{generalization cyclic}
Let $n$ be a positive odd integer and $C$ be a length $n$ linear cyclic code over $\F_4$ with the defining set $A$. If $A \cap -2A=\emptyset$, then there exists a $[[2(n-|A|),0,d]]$ quantum code (respectively a $[2(n-|A|),n-|A|,d]$ Hermitian self-dual linear code over $\F_4$) such that $d$ is even and $d \ge \min \{d(C^{\bot_h}),d(C)+1\}$.
\end{corollary}

\begin{proof}
We have $\dim(C^{\bot_h})=|A|$
and by Lemma $\ref{e lisonek}$, the condition $A \cap -2A=\emptyset$ implies that $C^{\bot_h} \subseteq C$.
Now the result follows from applying Theorem $\ref{general quantum construction}$ to the code $C^{\bot_h}$.
$\hfill\square$
\end{proof}

Next, we provide two new binary quantum codes that were obtained from two linear cyclic codes with $e=3$. 
\begin{example}
Let $n=141$ and $A=Z(2) \cup Z(3)\cup Z(10)$. Note that $-2A=Z(1) \cup Z(5)\cup Z(15)$ and therefore $A \cap -2A=\emptyset$. Moreover, $|A|=69$. Thus Corollary $\ref{generalization cyclic}$ implies the existence of a quantum code with parameters $[[144,0,d]]$, where $d$ is even and $d \ge \min \{d(C^{\bot_h}),d(C)+1\}$. Moreover, the minimum distance computation in Magma shows that $d(C^{\bot_h})\ge 20$ and $d(C)+1\geq 19$. Hence there exists a {\em new quantum code} with parameters $[[144,0,d\geq20]]$.
We note that a quantum code with these parameters
was found simultaneously by M.~Grassl \cite{Grassl3}
using different
methods.
 The previous best known binary quantum code with the same parameters had minimum distance $18$.
\end{example}

\begin{example}
Let $n=123$ and $A=Z(1) \cup Z(2)\cup Z(6) \cup Z(7) \cup Z(9) \cup Z(11)$. Note that $-2A=Z(43) \cup Z(23)\cup Z(3) \cup Z(19) \cup Z(18) \cup Z(14)$ and $A \cap -2A=\emptyset$. Moreover, $|A|=60$. Thus by Corollary $\ref{generalization cyclic}$ there exists a quantum code with parameters $[[126,0,d]]$, where $d$ is even and $d \ge \min \{d(C^{\bot_h}),d(C)+1\}$. Moreover, our Magma computation shows that $d(C^{\bot_h})\ge 22$ and $d(C)+1\geq 21$. The fact that $d$ is even and $d \geq 21$ implies that this quantum code has parameters $[[126,0,d \geq 22]]$ which is a {\em new quantum code}. The previous best quantum code with the same parameters had minimum distance $21$.
\end{example}
 
\section{Minimum distance bounds for cyclic codes using the fixed subcodes}

In general, computing the exact minimum distance for
general linear codes is NP-hard \cite{vardy} and very difficult for linear codes with large lengths and dimensions. 
In \cite{Joundan2}, the authors used the fixed subcode by the action of multipliers to find an upper bound (or even the exact value) for the minimum distance of certain linear cyclic codes. 
In this section, we extend the theory of fixed subcodes by the action of multipliers and determine a new minimum distance lower bound for linear cyclic codes over $\F_4$. 

In the rest of this section, we assume that $n$ is a positive odd integer.
Let $a$ be a positive integer such that $\gcd(n,a)=1$. Then $\mu_a$ acts naturally as a permutation on $\F_4^n$. In particular, 
let $\{e_i: 0 \le i \le n-1\}$ be the standard basis of $\F_4^n$. Then $\mu_a(e_i)=e_{ai}$ for each $0\le i \le n-1$, where 
vector indices are computed modulo $n$. For each $x=(x_0,x_1,\cdots,x_{n-1})\in \F_4^n$, we define $\mu_a(x)$ accordingly as
$\mu_a(x)=(y_0,y_1,\cdots,y_{n-1})$, where $y_i=x_{a^{-1}i}$ for any $0\le i\le n-1$. We denote the matrix representation of $\mu_a$ by $T_a$, that is, $T_ax^T=\mu_a(x)^T$ for each $x\in\F_4^n$.
Let $C$ be a length $n$ linear cyclic code over $\F_4$ with the defining set $A$. The code $\mu_a(C)$ is also a linear cyclic code over $\F_4$ and it has defining set $a^{-1}A$.

Now we formally define the fixed subcodes by the action of multipliers.
\begin{definition}
Let $C$ be a length $n$ linear cyclic code over $\F_4$. The 
linear space of all vectors $v \in C$ such that $\mu_{a}(v)=v$ is called the {fixed subcode} of $C$ under the action of $\mu_a$,
and it is denoted  $C_{a}$. 
\end{definition} 

The code $C_{a}$ is a subcode of $C \cap \mu_a(C)$ and it can be easily computed as $$C_a=C \cap \{ x\in\F_4^n : (T_a-I)x^T=0 \}.$$
 As before, for any integer $a$ such that $\gcd(n,a)=1$ we denote the multiplicative order of $a$ modulo $n$ by $\ord_n(a)$. 

The fixed subcodes of linear cyclic codes are especially important for us to bound the minimum distance of cyclic codes. First note that for each integer $a$ such that $\gcd(n,a)=1$, we have  $d(C)\le d(C_a)$ 
which is an upper bound for $d(C)$. Several of our minimum distance upper bounds in Table $\ref{Table2}$ are obtained 
in this way. Moreover, the next proposition provides a lower bound for the minimum distance of linear cyclic codes over $\F_4$ using the fixed subcodes.

\begin{proposition}\label{fix-subcode order 2}
Let $C\subseteq \F_4^n$ be a linear cyclic code with the defining set $A$ and $a$ be a positive integer such that $aA=A$.
\begin{enumerate}
\item If $\ord_n(a)=2$, then $d(C_{a})/2+1\leq d(C)$.
\item If $\ord_n(a)=i$, where $i>1$ is an odd integer, then $(d(C_{a})-1)/i+1\leq d(C)$.
\end{enumerate}
\end{proposition}

\begin{proof}
 Let $v=(v_0,v_1,\cdots,v_{n-1})$ be a minimum weight vector in $C$. Since $C$ is cyclic, without loss of generality, we assume that $v_{0}$ is non-zero. 

1. 
Suppose that $\ord_n(a)=2$.
Then $\mu_a(v+\mu_a(v))=v+\mu_a(v)$.
If $\mu_a(v)=v$, then $d(C)=d(C_{a})$ which completes the proof.  Now suppose $v+\mu_a(v) \neq0$.  
From $aA=A$ we get $\mu_a(C)=C$,
which implies that $v+\mu_{a}(v)$ is a non-zero element of $C_{a}$. Since both $v$ and $\mu_{a}(v)$ have the same coordinates in the $0$-th position we have $d(C_{a})\leq \wt(v+\mu_{a}(v))\leq 2d(C)-2$. Hence $d(C_{a})/2+1\leq d(C)$.

2. Let $\ord_n(a)=i$, where $i>1$ is an odd integer and $w=\sum_{j=0}^{i-1}\mu_{a^{j}}(v)$. Since $i$ is odd, $v_0=w_0$
and $w$ is a non-zero vector. Moreover $\mu_{a}(w)=w$ and therefore $w\in C_{a}$. Note also that $\wt(w)\le i\wt(v)-(i-1)=i(\wt(v)-1)+1$ since each $\mu_{a^{j}}(v)$ has $v_0$ in the $0$-th position. Thus, $d(C_{a})\leq \wt(w)\leq i(d(C)-1)+1$ which implies that $(d(C_{a})-1)/i+1\leq d(C)$.
$\hfill\square$
\end{proof}

Our computations in Magma computer algebra system \cite{magma} show that many linear cyclic codes
satisfying the conditions of Proposition $\ref{fix-subcode order 2}$ part 1 have the same minimum distance as their fixed subcode by an order~2 multiplier. 
In particular, for any $a \in \Z_n$ such that $\ord_n(a)=2$, we computed the minimum distance of all non-trivial length $n$ linear cyclic codes with $9<n<85$ over $\F_4$ satisfying  the conditions of Proposition \ref{fix-subcode order 2} part 1. Among 72417 non-trivial such linear cyclic codes, 70256 of them have the same minimum distance as their corresponding fixed subcode by the action of $\mu_{a}$. The equality rate is about $97\%$ for all these codes.
In general, determining when a linear cyclic code and its fixed subcode by $\mu_{-1}$ have the same minimum distance appears
to be an interesting and presumably a difficult question.

One application of Proposition $\ref{fix-subcode order 2}$ part 1 is provided below. In both of the following examples, the minimum distance of the fixed subcode was computed much faster, while the minimum distance computation for the original code required a much longer time. In particular, we found two new quantum codes after applying the minimum distance lower bound of Proposition $\ref{fix-subcode order 2}$. These codes are explained in detail below.

\begin{example}
Let $n=157$ and $C_o$ be the odd-like QR code over $\F_4$ with the defining set $Z(1)\cup Z(3)\cup Z(9)$. By Corollary $\ref{duadic construction 2}$ there exists a quantum code $Q$ with parameters $[[158,0,d]]$, where $d$ is even and $d \geq d(C_o)+1$. 

Next we use the result of Proposition $\ref{fix-subcode order 2}$ to find a lower bound for $d(C_o)$. Note that $\ord_n(4)=26$ and $4^{13}\equiv -1 \pmod{157}$, and we use the inequality $d( (C_{o})_{-1})/2+1\leq d(C_o)$, where $(C_{o})_{-1}$ is the fixed subcode of $C_o$ by the action of multiplier $\mu_{-1}$. 
Our computation done in Magma \cite{magma} shows that $d((C_{o})_{-1})=36$. Hence $d(C_o)\geq 19$ and the fact that $d$ is even shows that $Q$ has parameters $[[158,0,d\geq 20]]$. Thus $Q$ is a {\em new quantum code} with a better minimum distance in comparison with the previous best known code with the same length and dimension which had minimum distance $19$. 
\end{example}

\begin{example}
Let $n=181$ and $C_o$ be the odd-like QR code of length $n$ over $\F_4$ with the defining set $Z(1)$. Corollary $\ref{duadic construction 2}$ implies the existence of a quantum code $Q$ with parameters $[[182,0,d]]$, where $d$ is even and $d \geq d(C_o)+1$. 

Note that $\mu_{-1}(C_o)=C_o$ and our computations in Magma \cite{magma} show that the fixed subcode of $C_o$ under the action of multiplier $\mu_{-1}$ has the minimum distance $37$. Hence $d(C_o)\geq 19.5$ by Proposition $\ref{fix-subcode order 2}$ and the inequality $d \geq d(C_o)+1$ implies that $d\geq 20.5$. The fact that $d$ is even shows that $Q$ has parameters $[[182,0,d\geq 22]]$. Thus $Q$ is a {\em new quantum code;}  the previous best known quantum code had minimum distance $21$. 
\end{example}

Next we provide a connection between different fixed subcodes which also helps to relate the number of certain weight codewords in the original code and its fixed subcode.
\begin{theorem}
Let $C\subseteq \F_4^n$ be a linear cyclic code of length $n$ and $a$ be a positive integer such that $\ord_n(a)=p$ is prime. Let $A_t$ be the number of weight $t$ codewords in $C$ for any $0\le t \le n$. Then the following statements hold.
\begin{enumerate}
\item $C_{a}=C_{a^j}$ for any $1\le j \le p-1$.
\item Assume $\mu_a(C)=C$ and $0 \le t \le n$. Then either $C_a$ has a weight $t$ codeword or $p\mid A_t$. In particular, either $d(C)=d(C_{a})$ or $p\mid A_{d(C)}$.
\end{enumerate}
\end{theorem} 

\begin{proof} 
1. Let $1\le j \le p-1$. First note that if $v \in C_a$ then $\mu_{a^j}(v)=v$ and therefore $v\in C_{a^j}$.
Hence $C_a\subseteq C_{a^j}$. 
Next since $\gcd(p,j)=1$, we can find integers $b$ and $c$ such that $bp+cj=1$. If $u \in C_{a^j}$, then
\begin{equation*}\label{fixxx}
\mu_{a}(u)=(\mu_{a})^{bp+cj}(u)=(\mu_{a})^{bp}(\mu_{a})^{cj}(u)=(\mu_{a^{j}})^c(u)=u.\end{equation*}
Thus $C_{a^j}\subseteq C_{a}$ which implies that $C_{a}=C_{a^j}$ for any $1\le j \le p-1$.

2. If $A_t=0$ then the conclusion holds trivially.
Otherwise,
let $v$ be a weight $t$ vector in $C$. Then  
\begin{itemize}
\item either  $\mu_{a^j}(v)=v$ for all $1\le j \le p-1$
 \item or $v, \mu_a(v), \cdots, \mu_{a^{p-1}}(v)$ are all different weight $t$ codewords of $C$. 
 \end{itemize}
If the former happens for a weight $t$ codeword of $C$, then $C_a$ also has a weight $t$ vector. Otherwise, we can partition all the weight $t$ codewords of $C$ into sets of size $p$ in the form $\{v, \mu_a(v), \cdots, \mu_{a^{p-1}}(v)\}$. Thus $p\mid A_t$.  

The last part follows by choosing $t=d(C)$.
$\hfill\square$
\end{proof}

The previous result allows us to skip computing certain fixed subcodes in order to find bounds for the minimum distance of linear cyclic codes over $\F_4$. If the group $\Z_n^\ast$ is cyclic and $p \mid |\Z_n^\ast|$, then the code $C_a$ for any order $p$ element $a \in \Z_n^\ast$ is the same. Otherwise, there may exist
$a,b \in \Z_n^\ast$ of the same order  such that $C_a \neq C_b$.

\section{Numerical results}

The constructions given in Sections $\ref{construction section}$ and $\ref{section4}$ lead to many new quantum codes with minimum distances much higher than the previously best known codes.
In some cases the increase is by as much as~10. Overall, the computation is easiest when the near self-orthogonality parameter is $e=1$. In fact, in this case we arrive at the extended duadic codes. However, we also get three record-breaking quantum codes when $e=3$. So our construction goes beyond only the extended duadic codes.

Table $\ref{Table2}$ shows parameters of some good quantum codes. 
In the table, the first two columns show the length and the coset leaders (minimum elements of cyclotomic cosets contained
in the defining set)
of the cyclic code. The convention for the choice
of the primitive $n$-th root of unity for construction
of cyclic codes was explained in a remark
in Section~\ref{subsec-duadic}.
The third column records whether the original code is a QR code, duadic code, or 
some general cyclic code. This is indicated
  with QR, D, and C respectively.

In Table $\ref{Table2}$, 
we used the probable minimum distances provided in \cite[Section 6.5]{Huffman} for binary duadic codes of lengths $217$, $233$, and $239$. The binary and quaternary generator polynomials are the same for these three codes. The probable minimum distance $d$ for each of these three values is denoted by $d^{ap}$ in the table.
All the other minimum distances given in the table are the true minimum distance obtained from some
combination of the minimum distance lower bound of the related construction
and/or computation by the built-in minimum distance function in computer algebra system Magma \cite{magma}, or a reference for
the minimum distance is provided in the source column.

When the exact value of the minimum distance is not known, its lower and upper bounds are separated by a dash. Some of the minimum distance upper bounds presented in Table $\ref{Table2}$ are computed using 
 Magma \cite{magma} functions for attacking the McEliece cryptosystem.  

The ``source'' column in the table provides information about the result used to construct the quantum code. We label theorems, propositions, and corollaries by their first letter in this column. 

Finally, the PMD column shows the minimum distance of previous best known quantum code of the same length and dimension as shown in \cite{Grassl}. In cases where our code listed in Table $\ref{Table2}$ has a strictly higher minimum distance than the previous best known quantum code, we list the distance of our code in boldface in the parameters column.

%

\begin{table}
\vspace{-1.75cm}
\begin{center}
\begin{tabular}{ |p{1.4 cm}|p{2.5 cm}|p{0.8 cm}|p{2.6 cm}| p{1.9 cm}|p{0.9 cm}|}
\hline
 Length & Coset Leaders& Type &Parameters & Source &PMD\\
 \hline
 $n=5$ & $ 1$ & QR & $[[6,0,4]]$& T\ref{duadic-quantum}&4\\
 $n=7$ & $1$ & QR &$[[8,0,4]]$&T\ref{duadic-quantum} & 4\\
 $n=13$ & $ 1$ & QR & $[[14,0,6]]$&T\ref{duadic-quantum}& 6\\
$n=17$ & $ 1,3$ &D& $[[18,0,8]]$& T\ref{duadic-quantum}& 8 \\
 $n=23$ & $1$ & QR & $[[24,0,8]]$& T\ref{duadic-quantum}&8\\
 $n=29$ & $ 1$ & QR&  $[[30,0,12]]$&T\ref{duadic-quantum}& 12\\
 $n=37$ & $ 1$ & QR& $[[38,0,12]]$&T\ref{duadic-quantum}& 12\\
 $n=41$ & $ 1,3$ & D& $[[42,0,12]]$&T\ref{duadic-quantum}& 12\\
 $n=53$ & $ 1$ & QR& $[[54,0,16]]$&T\ref{duadic-quantum}& 16\\
 $n=61$ & $ 1$ &QR&$[[62,0,18]]$&T\ref{duadic-quantum}& 18\\
 $n=93$ & $ 1, 5, 9, 13, 17, 23,$ & C& $[[96,0,{\bf22}]]$&C\ref{generalization cyclic} (e=3)& 20\\
  & 33, 34, 45&&&&\\
 $n=101$ & $ 1$ & QR& $[[102,0,22]]$&T\ref{duadic-quantum}& 22\\
 $n=103$ & $1$ &QR & $[[104,0,20]]$& \cite{Huffman} \& T\ref{binary-quaternary}&20 \\
 $n=113$ & $1,3,5,9$ & D&$[[114,0, 20-26]]$&T\ref{duadic-quantum}& 24\\
 $n=119$ & $1,2,3,6,7,21,51$ &D &$[[120,0, {\bf 20}]]$&T\ref{duadic-quantum}& 18\\
  $n=123$ & $1,2,6,7,9,11$ & C&$[[126,0,{\bf 22}-24]]$& C\ref{generalization cyclic} (e=3) & 21\\
 &19, 21, 31, 47&&&&\\
 $n=137$ & $1,3$ & D& $[[138,0,{\bf 22}]]$ & T\ref{duadic-quantum}& 18\\
   $n=141$ & $2,3,10$ & C&$[[144,0, 20]]$& C\ref{generalization cyclic} (e=3)& 20\\
 $n=145$ & $1,3,5,7,11,29$ &D& $[[146,0,18-32]]$ &T\ref{duadic-quantum} & 21\\
 $n=149$ & $1$ &QR& $[[150,0,18-30]]$ &T\ref{duadic-quantum} & 23\\
 $n=151$ & $1,3,7,11,15$ &D& $[[152,0,{\bf 24}]]$ &\cite{Gaborit} \& T\ref{binary-quaternary}&18\\
 $n=155$ & $ 1, 6, 13, 14, 18, $ &D & $[[156,0,18-24]]$ &T\ref{duadic-quantum}& 19 \\
 &$ 25, 29,35, 62, 75 $&&&&\\
 $n=157$ & $1,3,9$& QR& $[[158,0,{\bf20}-36]]$ & T\ref{duadic-quantum} \& P\ref{fix-subcode order 2} & 19 \\
 $n=167$ & $1$ & QR& $[[168,0,{\bf24}]]$ &\cite{Truong} \& T\ref{binary-quaternary}& 20\\
 $n=173$ & $1$ & QR& $[[174,0,20-36]]$ &T\ref{duadic-quantum} \& P\ref{fix-subcode order 2} & 21 \\
 $n=181$ & $1$& QR& $[[182,0,22-38]]$ & T\ref{duadic-quantum} \& P\ref{fix-subcode order 2} & 23 \\
 $n=185$ & $ 2, 3, 9, 10, 19, 74 $& D& $[[186,0,18-36]]$ & T\ref{duadic-quantum} & 25 \\
 $n=191$ & $1$ & QR & $[[192,0,{\bf28}]]$ & \cite{SuWK} \& T\ref{binary-quaternary}&22\\
 $n=193$ & $1,5$ &D & $[[194,0,20-42]]$ &T\ref{duadic-quantum}& 29 \\
 $n=197$ & $1$& QR& $[[198,0,22-40]]$ &T\ref{duadic-quantum} \& P\ref{fix-subcode order 2} & 31\\
 $n=199$ & $1$ & QR & $[[200,0,{\bf32}]]$&\cite{SuWK} \& T\ref{binary-quaternary}& 22\\
 $n=203$ & $2,3,7,29$ &D & $[[204,0,14-24]]$&T\ref{duadic-quantum}& 22\\
 $n=205$ & $ 2, 5, 6, 7, 18, 21, $ &D & $[[206,0,18-44]]$&T\ref{duadic-quantum} & 23\\
 &$22, 30, 31, 34, 82 $&&&&\\
 $n=221$ & $ 3, 13, 18, 22, 23,31$ &D & $[[222,0,16-40]]$&T\ref{duadic-quantum} &31\\
 &$ 33, 34, 55, 78, 93 $&&&&\\
 $n=223$ & $1,9,19$ & QR & $[[224,0,{\bf 32}]]$&\cite{Joundan} \& T\ref{binary-quaternary} &21\\
 $n=229$ & $1,5,6$& D& $[[230,0,18-54]]$ &T\ref{duadic-quantum}& 22 \\
 $n=233$ & $1,3,7,27$ &D& $[[234,0,{\bf30^{ap}}]]$ &\cite{Huffman} \& T\ref{binary-quaternary} &20 \\
 $n=235$ & $1,2,5,47$ &D& $[[236,0,14-24]]$ &T\ref{duadic-quantum} &20\\
 $n=239$ & $1$ & QR & $[[240,0,{\bf 32^{ap}}]]$ & \cite{Huffman} \& T\ref{binary-quaternary} &20\\
 $n=241$ & $1,3,5,7,9,11,$& D& $[[242,0,14-56]]$ &T\ref{duadic-quantum}& 20\\
 &$13,21,25,35$&&&&\\
\hline
\end{tabular}
\vspace{2mm}
\end{center}
\caption{Parameters of some good $0$-dimensional quantum codes.} 
\label{Table2}
\end{table}

It should be noted that we can apply the secondary construction given in Theorem $\ref{secondary}$ part 2 to the codes listed in Table \ref{Table2} and produce many more record-breaking codes. For instance:

\begin{itemize}
	
	\item the quantum code $[[240,0,32]]$ generates $9$ new quantum codes with parameters $[[240-i,0,32-i]]$ for each $1\le i\le 9$.
	
	\item the quantum code $[[234,0,30]]$ generates $7$ new quantum codes with parameters $[[234-i,0,30-i]]$ for each $1\le i\le 7$.
	
\end{itemize}

The codes obtained
from secondary constructions
are not listed in Table~$\ref{Table2}$.


\section*{Acknowledgement}
The authors would like to thank Markus Grassl for many
interesting discussions and for sharing
the recent preprint \cite{Grassl3}.
This work was supported by 
the Natural Sciences and Engineering Research Council of Canada (NSERC, Project No.\ RGPIN-2015-06250 and RGPIN-2022-04526),

\bibliographystyle{abbrv}
\bibliography{DL-self-dual}

\end{document}